\documentclass{sigplanconf}

\usepackage{stmaryrd} 
\usepackage{amsthm, amssymb} 
\usepackage{verbatim}
\usepackage{hyperref}

\newcommand{\toy}{\mathcal{TOY}} 
\newcommand{\REAL}{\mathbb{R}} 

\newcommand{\qdom}{\mathcal{D}} 
\newcommand{\dqdom}{D \setminus \{\bot\}} 

\newcommand{\B}{\mathcal{B}}
\newcommand{\U}{\mathcal{U}}
\newcommand{\W}{\mathcal{W}}
\newcommand{\Set}{\mathcal{S}}

\newcommand{\simrel}{\mathcal{R}} 

\newcommand{\SL}{\mathcal{T}} 
\newcommand{\GL}{\mathcal{G}} 

\newcommand{\qlp}[1]{QLP({#1})} 
\newcommand{\sqlp}[2]{SQLP({#1,#2})} 

\newcommand{\trans}[2]{S_{#1}(#2)} 

\newcommand{\extended}[2]{H_{#1}(#2)} 
\newcommand{\abstracted}[2]{{#1}_{#2}} 

\newcommand{\MAProg}{\mathcal{P}_{E, \simrel}}

\newcommand{\diff}{~{\Longleftrightarrow_{\mathrm{def}}}~} 
\newcommand{\union}{\bigcup} 
\newcommand{\inter}{\bigcap} 
\newcommand{\supr}{\bigsqcup} 
\newcommand{\infi}{\bigsqcap} 
\newcommand{\Dentail}{~{\succcurlyeq_{\qdom}}~} 
\newcommand{\DentailSim}{~{\succcurlyeq_{(\simrel,\qdom)}}~} 

\newcommand{\Prog}{\mathcal{P}} 

\newcommand{\Var}{\mathcal V\!ar} 
\newcommand{\War}{\mathcal W\!ar} 
\newcommand{\Tp}{\mathrm{T}_{\Prog}} 
\newcommand{\Mp}{\mathcal{M}_{\Prog}} 
\newcommand{\M}[1]{\mathcal{M}_{#1}} 
\newcommand{\Atz}{\mathrm{At}_{\Sigma}} 
\newcommand{\QAtz}{\mathrm{At}_{\Sigma}(\qdom)} 
\newcommand{\intd}{\mathrm{Int}_{\Sigma}(\qdom)} 
\newcommand{\intrd}{\mathrm{Int}_{\Sigma}(\simrel,\qdom)} 
\newcommand{\I}{\mathcal{I}} 
\newcommand{\at}[2]{#1\,\sharp\,#2} 
\newcommand{\ats}[1]{\overline{#1}} 
\newcommand{\qgets}[1]{\gets\!#1\!-} 
\newcommand{\sep}{~{\talloblong}~} 
\newcommand{\linear}[1]{#1_\ell} 

\newcommand{\qhl}[1]{QHL({#1})} 
\newcommand{\sqhl}[2]{SQHL({#1,#2})} 
\newcommand{\sqhlx}[3]{\vdash_{\mathrm{SQHL}(#1,#2)}^{#3}} 
\newcommand{\sqhlrd}{\sqhlx{\simrel}{\qdom}{}} 
\newcommand{\sqhlrdn}[1]{\sqhlx{\simrel}{\qdom}{#1}} 
\newcommand{\qhlx}[2]{\vdash_{\mathrm{QHL}(#1)}^{#2}} 
\newcommand{\qhld}{\qhlx{\qdom}{}} 
\newcommand{\qhldn}[1]{\qhlx{\qdom}{#1}} 

\newcommand{\MALP}{{\em MALP }}
\newcommand{\malp}{{\em MALP}}

\theoremstyle{definition}
\newtheorem{definition}{Definition}

\theoremstyle{plain}
\newtheorem{lemma}{Lemma}
\newtheorem{proposition}{Proposition}
\newtheorem{theorem}{Theorem}
\newtheorem{example}{Example}

\begin{document}

\conferenceinfo{PPDP'08,} {July 15--17, 2008, Valencia, Spain.}
\CopyrightYear{2008}
\copyrightdata{978-1-60558-117-0/08/07}

\titlebanner{banner above paper title}        
\preprintfooter{short description of paper}   

\title{Similarity-based Reasoning in Qualified Logic Programming}
\subtitle{Revised Edition}

\authorinfo{Rafael Caballero \and Mario Rodr\'iguez-Artalejo \and Carlos A. Romero-D\'iaz}
           {Departamento de Sistemas Inform\'aticos y Computaci\'on \\ Universidad Complutense de Madrid, Spain}
           {$\{$rafa,mario$\}$@sip.ucm.es, cromdia@fdi.ucm.es}

\maketitle

\begin{abstract}
{\em Similarity-based Logic Programming} (briefly, $SLP$) has been proposed to enhance the $LP$ paradigm with a kind of approximate reasoning which supports flexible information retrieval applications. This approach uses a fuzzy similarity relation $\simrel$ between symbols in the program's signature,  while keeping the syntax for program  clauses as in classical $LP$. Another recent proposal is the  $\qlp{\qdom}$ scheme for {\em Qualified Logic Programming}, an extension of the $LP$ paradigm which supports approximate reasoning and more. This approach uses annotated program clauses and a parametrically given domain $\qdom$  whose elements  qualify logical assertions by measuring their closeness to various users' expectations. In this paper we propose  a more expressive scheme $\sqlp{\simrel}{\qdom}$ which subsumes both $SLP$ and $\qlp{\qdom}$ as particular cases. We also show that $\sqlp{\simrel}{\qdom}$ programs can be transformed into semantically equivalent $\qlp{\qdom}$ programs.  As a consequence, existing $\qlp{\qdom}$ implementations can be used to give efficient support for similarity-based reasoning.
\end{abstract}




\section{Introduction} \label{Introduction}

The historical evolution of the research on uncertainty in {\em Logic Programming} ($LP$) has been described in a recent recollection by V. S. Subrahmanian \cite{Sub07}. Early approaches include the quantitative  treatment of uncertainty in the spirit of fuzzy logic, as in van Emden's classical paper \cite{VE86} and two subsequent papers by Subrahmanian \cite{Sub87,Sub88}. The main contribution of \cite{VE86} was a rigorous declarative semantics for a $LP$ language with program clauses  of the form $A \qgets{d} \ats{B}$, where the head $A$ is an atom, the body $\ats{B}$ is a conjunction of atoms, and the so-called {\em attenuation} factor $d \in (0,1]$ attached to the clause's implication is used to propagate
to the head the certainty factor $d \times b$, where $b$ is the minimum of the certainty factors $d_i \in (0,1]$ previously computed for the various atoms occurring in the body. The papers \cite{Sub87,Sub88} proposed to use a special lattice $\SL$ in place of the lattice of the real numbers in the interval $[0,1]$ under their natural ordering. $\SL$ includes two isomorphic copies of $[0,1]$ whose elements are incomparable under $\SL$'s ordering and can be used separately to represent degrees of {\em truth} and {\em falsity}, respectively, thus enabling a simple treatment of
negation. Other main contributions of \cite{Sub87,Sub88} were the introduction of annotated program clauses and goals (later generalized to a much more expressive framework in \cite{KS92}), as well as goal solving procedures more convenient and powerful than those given in \cite{VE86}.

A more recent line of research is {\em Similarity-based Logic Programming}  (briefly, $SLP$) as presented in \cite{Ses02} and previous related works such as \cite{AF99,GS99,FGS00,Ses01}. This approach also uses the lattice $[0,1]$ to deal with uncertainty in the spirit of fuzzy logic. In contrast to approaches based on annotated clauses, programs in $SLP$ are just sets of definite Horn clauses as in classical $LP$. However, a {\em similarity relation} $\simrel$ (roughly, the fuzzy analog of an equivalence relation) between predicate and function symbols is used to enable the
unification terms that would be not unifiable in the classical sense, measured by some degree $\lambda \in (0,1]$. There are different proposals for the operational semantics of $SLP$ programs. One possibility is to apply  classical $SLD$ resolution w.r.t. a transformation of the original program \cite{GS99,Ses01,Ses02}. Alternatively, a $\simrel$-based $SLD$-resolution procedure relying on $\simrel$-unification can be applied w.r.t. to the original program, as proposed in \cite{Ses02}. Propositions 7.1 and 7.2 in \cite{Ses02} state a correspondence between the answers computed by
$\simrel$-based $SLD$ resolution w.r.t. a given logic program $\Prog$ and the answers computed by classical $SLD$ resolution w.r.t. the two transformed programs $\extended{\lambda}{\Prog}$ (built by adding to $\Prog$ new clauses $\simrel$-similar to those in $\Prog$ up to the degree $\lambda \in (0,1]$) and $\abstracted{\Prog}{\lambda}$ (built by replacing all the function and predicate symbols in $\Prog$ by new symbols that represent
equivalence classes modulo $\simrel$-similarity up to $\lambda$). The $SiLog$ system  \cite{LSS04} has been developed to implement $SLP$ and to support applications related to flexible information retrieval from the web.

The aim of the present paper is to show that similarity-based reasoning can be expressed in $\qlp{\qdom}$, a programming scheme for {\em Qualified} $LP$ over a parametrically given {\em Qualification Domain} $\qdom$ recently presented in \cite{RR08} as a generalization and improvement of the classical approach by van Emden \cite{VE86} to $Quantitative$ $LP$. Qualification domains are lattices satisfying certain natural axioms. They include the lattice $[0,1]$ used both in \cite{VE86} and in \cite{Ses02}, as well as other lattices whose elements can be used to qualify logical
assertions by measuring their closeness to different kinds of users' expectations. Programs in $\qlp{\qdom}$ use $\qdom$-attenuated clauses of the form $A \qgets{d} \ats{B}$ where $A$ is an atom, $\ats{B}$ a finite conjunction of atoms and $d \in \dqdom$ is the {\em attenuation value} attached to the clause's implication, used to propagate to the head the {\em qualification value} $d \circ b$, where $b$ is the infimum in $\qdom$ of the qualification values $d_i \in \dqdom$ previously computed for the various atoms occurring in the body, and  $\circ$ is an {\em attenuation operator} coming
with $\qdom$. As reported in \cite{RR08,RR08TR}, the classical results in $LP$ concerning the existence of least Herbrand models of programs and the soundness and completeness of the  $SLD$ resolution procedure (see e.g.\cite{VEK76,AVE82,Apt90}) have been extended to the $\qlp{\qdom}$ scheme, and potentially useful instances of the scheme have been implemented on top of the {\em Constraint Functional Logic Programming} ($CFLP$) system $\toy$ \cite{toy}.

The results presented in this paper can be summarized as follows: we consider generalized similarity relations over a set $S$ as mappings $\simrel: S \times S \to D$ taking values in the carrier set $D$ of an arbitrarily  given qualification domain $\qdom$, and we  extend $\qlp{\qdom}$ to a more expressive scheme $\sqlp{\simrel}{\qdom}$ with two parameters for programming modulo $\simrel$-similarity with $\qdom$-attenuated Horn clauses. We present a declarative semantics for $\sqlp{\simrel}{\qdom}$ and a program transformation mapping each $\sqlp{\simrel}{\qdom}$ program $\Prog$ into a
$\qlp{\qdom}$ program $\trans{\simrel}{\Prog}$ whose least Herbrand model corresponds to that of $\Prog$. Roughly, $\trans{\simrel}{\Prog}$ is built adding to $\Prog$ new clauses obtained from the original clauses in $\Prog$ by computing various new heads $\simrel$-similar to a linearized version of the original head, adding also $\simrel$-similarity conditions $X_i \sim X_j$  to the body and suitable clauses for  the  new predicate $\sim$ to
emulate $\simrel$-based unification. Thanks to the $\trans{\simrel}{\Prog}$ transformation, the sound and complete procedure for solving goals in $\qlp{\qdom}$ by $\qdom$-qualified $SLD$ resolution and its implementation in the $\toy$ system \cite{RR08} can be used to implement $\sqlp{\simrel}{\qdom}$ computations, including as a particular case $SLP$ computations in the sense of \cite{Ses02}.

Another recent proposal for reducing the $SLP$ approach in \cite{Ses02} to a fuzzy $LP$ paradigm can be found in \cite{MOV04}, a  paper which relies on  the multi-adjoint framework for Logic Programming ({\MALP} for short) previously proposed in \cite{MOV01a,MOV01b}. {\MALP} is a quite general framework supporting $LP$ with {\em weighted program rules} over different multi-adjoint lattices, each of which provides a particular choice of operators for implication, conjunction  and aggregation of atoms in rule bodies. In comparison to the $\qlp{\qdom}$ scheme, the multi-adjoint framework differs in motivation and scope. Multi-adjoint lattices and qualification domains are  two different classes of algebraic structures. Concerning declarative and
operational semantics, there are also some significant differences between $\qlp{\qdom}$ and \malp. In particular, \MALP's goal solving procedure relies on a costly computation of {\em reductant clauses}, a technique borrowed from \cite{KS92} which can be avoided in $\qlp{\qdom}$, as discussed in the concluding section of \cite{RR08}.

In spite of these differences, the results in \cite{MOV04} concerning the emulation of similarity-based  can be compared to those in the present paper. Theorem 24 in \cite{MOV04} shows that  every classical logic  program $\Prog$ can be transformed into a \MALP program $\MAProg$ which can be executed using only syntactical unification and emulates the successful computations of $\Prog$ using the $SLD$ resolution with $\simrel$-based unification
introduced  in \cite{Ses02}. $\MAProg$ works over a particular multi-adjoint lattice $\GL$ with carrier set $[0,1]$ and implication and conjunction operators chosen according to the so-called G\"{o}del's semantics \cite{Voj01}. $\MAProg$ also introduces clauses for a binary predicate $\sim$ which emulates $\simrel$-based unification, as in our transformation $\trans{\simrel}{\Prog}$. Nevertheless, $\trans{\simrel}{\Prog}$   is defined for a more general class of programs and uses the $\simrel$-similarity predicate $\sim$ only if the source program $\Prog$ has some clause whose head is non-linear. More detailed comparisons between the program transformations $\trans{\simrel}{\Prog}$,  $\extended{\lambda}{\Prog}$, $\abstracted{\Prog}{\lambda}$ and $\MAProg$ will be given in Subsection \ref{sec:RA}.

The rest of the paper is structured as follows: In Section \ref{Domains} we recall the qualification domains $\qdom$  first introduced in \cite{RR08} and we define similarity relations $\simrel$ over an arbitrary qualification domain. In Section \ref{Language} we recall the scheme $\qlp{\qdom}$ and we introduce its extension $\sqlp{\simrel}{\qdom}$ with its declarative semantics, given by a logical calculus which characterizes the least
Herbrand model $\M{\Prog}$ of each $\sqlp{\simrel}{\qdom}$ program $\Prog$. In Section \ref{Reduction} we define the transformation $\trans{\simrel}{\Prog}$ of any given $\sqlp{\simrel}{\qdom}$ program $\Prog$ into a $\qlp{\qdom}$ program $\trans{\simrel}{\Prog}$ such that $\M{\trans{\simrel}{\Prog}} = \M{\Prog}$, we give some comparisons to previously known program transformations, and we illustrate the application of
$\trans{\simrel}{\Prog}$ to similarity-based computation by means of a simple example. Finally, in Section \ref{Conclusions} we summarize conclusions and comparisons to related work and we point to planned lines of future work.

\section{Qualification Domains and Similarity Relations} \label{Domains}

\subsection{Qualification Domains} \label{QD}
{\em Qualification Domains}  were introduced in \cite{RR08} with the aim of using their elements to qualify logical assertions in different ways. In this subsection we recall their axiomatic definition and some significant examples.

\begin{definition}  \label{defQD}
A  {\em Qualification Domain} is any structure $\qdom = \langle D, \sqsubseteq,$ $\bot, \top, \circ \rangle$ verifying the following requirements:
\begin{enumerate}
    \item $\langle D, \sqsubseteq, \bot, \top \rangle$ is a lattice with extreme points $\bot$ and $\top$ w.r.t. the partial ordering $\sqsubseteq$. For given elements  $d, e \in D$, we  write $d\, \sqcap\, e$ for the {\em greatest lower bound} ($glb$) of $d$ and $e$ and $d\, \sqcup\, e$ for the {\em least upper bound} ($lub$) of $d$ and $e$. We also write $d \sqsubset e$ as abbreviation for $d \sqsubseteq e\, \land\, d \neq e$.
    \item $\circ : D \times D \rightarrow D$, called {\em attenuation operation}, verifies the following axioms:
        \begin{enumerate}
            \item $\circ$ is associative, commutative and monotonic w.r.t. $\sqsubseteq$.
            \item $\forall d \in D:\, d \circ \top = d$.
            \item $\forall d \in D:\, d \circ \bot = \bot$.
            \item $\forall d, e \in D \setminus \{\bot,\top\}:\, d \circ e\, \sqsubset\, e$.
            \item $\forall d, e_1, e_2 \in D:\, d  \circ (e_1 \sqcap e_2) = d \circ e_1\, \sqcap\, d \circ e_2$. \qed
        \end{enumerate}
\end{enumerate}
\end{definition}

In the rest of the paper, $\qdom$ will generally denote an arbitrary qualification domain. For any finite $S = \{e_{1}, e_{2}, \ldots, e_{n}\} \subseteq D$, the $glb$ of $S$ (noted as $\bigsqcap S$) exists and can be computed as $e_{1} \sqcap e_{2} \sqcap \cdots \sqcap e_{n}$ (which reduces to $\top$ in the case $n = 0$). As an easy consequence of the axioms, one gets the identity $d \circ \bigsqcap S =  \bigsqcap \{d \circ e \mid e \in S\}$.  The $\qlp{\qdom}$ scheme presented in \cite{RR08} supports $LP$ over a parametrically given qualification domain $\qdom$.

\begin{example} \label{someDomains}
Some examples of  qualification domains are presented below. Their intended use for qualifying logical assertions will become more clear in Subsection \ref{QLP}.
\begin{enumerate}
\item
$\B = (\{0,1\}, \le, 0, 1, \land)$, where $0$ and $1$ stand for the two classical truth values  \emph{false} and \emph{true}, $\leq$ is the usual numerical ordering over $\{0,1\}$, and $\land$ stands for the classical conjunction operation over $\{0,1\}$. Attaching $1$ to  an atomic formula  $A$ is intended to qualify $A$ as `true' in the sense of classical $LP$.
\item
$\U = (\mbox{U}, \leq, 0, 1,\times)$, where $\mbox{U} = [0,1] = \{d \in \REAL \mid 0 \le d \le 1\}$, $\le$ is the usual numerical ordering, and $\times$ is the multiplication operation. In this domain,  the top element $\top$ is $1$ and the greatest lower bound $\bigsqcap S$ of a finite $S \subseteq \mbox{U}$ is the minimum value min(S), which is $1$ if $S = \emptyset$. Attaching an  element $c \in \mbox{U} \setminus \{0\}$  to an atomic formula  $A$ is intended to qualify $A$ as `true with certainty  degree $c$' in the spirit of fuzzy logic, as done in the classical paper \cite{VE86} by van Emden. The computation of qualifications $c$ as certainty degrees in $\U$ is due to the interpretation of $\sqcap$ as $min$ and $\circ$ as $\times$.
\item
$\W = (\mbox{P}, \ge, \infty, 0, +)$, where $\mbox{P} = [0,\infty] = \{d \in \REAL \cup \{\infty\} \mid d \ge 0\}$, $\geq$ is the reverse of the usual numerical ordering (with $\infty \ge d$ for any $d \in \mbox{P}$), and $+$ is the addition operation (with $\infty + d = d + \infty = \infty$ for any $d \in \mbox{P}$). In this domain,  the top element $\top$ is $0$ and the greatest lower bound $\bigsqcap S$ of a finite $S \subseteq \mbox{P}$ is the maximum value max(S), which is $0$ if $S = \emptyset$. Attaching an element $d \in \mbox{P} \setminus \{\infty\}$ to an atomic formula $A$ is intended to qualify $A$ as `true with weighted proof depth $d$'. The computation of qualifications $d$ as weighted proof depths in $\W$ is due to the interpretation  of $\sqcap$ as $max$ and $\circ$ as $+$.
\item
Given 2 qualification domains  $\qdom_i=\langle D_i, \sqsubseteq_i, \bot_i, \top_i, \circ_i \rangle$ ($i \in \{1, 2\}$), their {\em cartesian product} $\qdom_1 \times \qdom_2$ is $\qdom =_{\mathrm{def}} \langle D, \sqsubseteq, \bot, \top, \circ \rangle$, where  $D =_{\mathrm{def}} D_1 \times D_2$, the partial ordering $\sqsubseteq$ is defined as $(d_1,d_2) \sqsubseteq (e_1,e_2) \diff d_1 \sqsubseteq_1 e_1$ and $d_2 \sqsubseteq_2 e_2$, $\bot =_{\mathrm{def}} (\bot_1, \bot_2)$, $\top =_{\mathrm{def}} (\top_1, \top_2)$, and the attenuation operator $\circ$ is defined as
$(d_1,d_2) \circ (e_1,e_2) =_{\mathrm{def}} (d_1 \circ_1 e_1, d_2 \circ_2 e_2)$. The product of two given qualification domains is always another  qualification domain, as proved in \cite{RR08}. Intuitively, each value $(d_1,d_2)$ belonging to $\qdom_1 \times \qdom_2$ imposes the qualification $d_1$ {\em and also} the qualification $d_2$. For instance, values $(c,d)$ belonging to $\U \times \W$ impose two qualifications, namely: a certainty degree greater or equal than $c$ and a weighted proof depth less or equal than $d$. \qed
\end{enumerate}
\end{example}

For technical reasons that will become apparent in Section \ref{Reduction}, we consider the two structures $\U'$ resp. $\W'$ defined analogously to  $\U$ resp. $\W$, except that $\circ$ behaves as $min$ in $\U'$ and as $max$ in $\W'$. Note that almost all the axioms for qualification domains enumerated in Definition \ref{defQD} hold in $\U'$ and $\W'$, except that axiom $2. (d)$ holds only in the relaxed form
$\forall d, e \in D:\, d \circ e\, \sqsubseteq\, e$. Therefore, we will refer to $\U'$ and $\W'$ as {\em quasi} qualification domains.

\subsection{Similarity relations} \label{SR}

{\em Similarity relations} over a given set $S$  have been defined in \cite{Ses02} and related literature as mappings $\simrel : S \times S \to [0,1]$ that satisfy three axioms analogous to those required for classical equivalence relations. Each value $\simrel(x,y)$ computed by a similarity relation $\simrel$ is called the {\em similarity  degree} between $x$ and $y$. In this paper we use a natural extension of the definition given in \cite{Ses02}, allowing elements of an arbitrary qualification domain $\qdom$  to serve as similarity degrees. As in \cite{Ses02}, we  are especially interested in similarity relations over sets $S$ whose elements are variables and symbols of a given signature.

\begin{definition} \label{defSR} Let a qualification domain $\qdom$ with carrier set $D$ and a set $S$ be given.
\begin{enumerate}
\item A {\em $\qdom$-valued similarity relation} over $S$ is any mapping  $\simrel : S \times S \to D$ such that the three following axioms hold for all $x,y,z \in S$:
    \begin{enumerate}
        \item {\em Reflexivity:} $\simrel(x,x) = \top$.
        \item {\em Symmetry:} $\simrel(x,y) = \simrel(y,x)$.
        \item {\em Transitivity:} $\simrel(x,z) \sqsupseteq \simrel(x,y)\, \sqcap\, \simrel(y,z)$.
    \end{enumerate}
\item The mapping $\simrel : S \times S \to D$ defined as $\simrel(x,x) = \top$ for all $x \in D$ and $\simrel(x,y) = \bot$ for all $x,y \in D$, $x \neq y$ is trivially a $\qdom$-valued similarity relation called the \emph{identity}.
\item A $\qdom$-valued similarity relation $\simrel$ over $S$ is called {\em admissible} iff $S = \Var\, \cup\, CS \cup\, PS$ (where the three mutually disjoint sets $\Var$, $CS$ and $PS$ stand for a countably infinite collection of {\em variables}, a set of {\em constructor symbols} and a set of {\em predicate symbols}, respectively) and the two following requirements are satisfied:
    \begin{enumerate}
        \item $\simrel$ restricted to $\Var$ behaves as the identity, i.e. $\simrel(X,X) = \top$ for all $X \in \Var$ and $\simrel(X,Y) = \bot$ for all $X, Y \in \Var$, $X \neq Y$.
        \item $\simrel(x,y) \neq \bot$ holds only if some of the following three cases holds $x, y$: either $x, y \in \Var$ are both the same variable; or else $x, y \in CS$ are constructor symbols with the same arity; or else $x, y \in PS$ are predicate  symbols with the same arity. \qed
    \end{enumerate}
\end{enumerate}
\end{definition}

The similarity degrees computed by a $\qdom$-valued similarity relation must be interpreted w.r.t. the intended role of $\qdom$-elements as qualification values. For example, let  $\simrel$ be an admissible similarity relation, and let $c, d \in CS$ be two nullary constructor symbols (i.e., constants). If $\simrel$ is $\U$-valued, then $\simrel(c,d)$ can be interpreted as a {\em certainty degree} for the assertion that $c$ and $d$ are similar. On the other hand, if $\simrel$ is $\W$-valued, then $\simrel(c,d)$ can be interpreted as a {\em cost} to be paid for $c$ to play the role of $d$. These two views are coherent with the different interpretations of the operators $\sqcap$ and $\circ$ in $\U$ and $\W$, respectively.

In the rest of the paper we assume that any admissible similarity relation $\simrel$ can be extended to act over terms, atoms and clauses. The extension, also called $\simrel$, can be recursively defined as in \cite{Ses02}. The following definition specifies the extension of $\simrel$ acting over terms. The case of atoms and clauses is analogous.

\begin{definition} ($\simrel$ acting over terms). \label{def:ER}
\begin{enumerate}
    \item For $X \in \Var$ and for any term $t$ different from $X$:\\
    $\simrel(X,X) = \top$ and $\simrel(X,t) = \simrel(t,X) = \bot$.

    \item For $c, c' \in CS$ with different arities $n$, $m$:\\
    $\simrel(c(t_1, \ldots, t_n), c'(t_1', \ldots, t_m')) = \bot$.

    \item For $c, c' \in CS$ with the same arity $n$:\\
    $\simrel(c(t_1, \ldots, t_n), c'(t_1', \ldots, t_n')) = \simrel(c,c') \sqcap \simrel(t_1,t_1') \sqcap \ldots \sqcap \simrel(t_n,t_n')$.
\end{enumerate}
\end{definition}

\section{Similarity-based Qualified Logic Programming}\label{Language}

In this section  we extend our previous scheme $\qlp{\qdom}$   to a more expressive scheme called \emph{Similarity-based Qualified Logic Programming} over $(\simrel,\qdom)$ --abbreviated as $\sqlp{\simrel}{\qdom}$-- which supports both qualification over $\qdom$ in the sense of \cite{RR08} and $\simrel$-based similarity in the sense of \cite{Ses02} and related research. Subsection \ref{QLP} presents a quick review of the main results concerning syntax and declarative semantics of $\qlp{\qdom}$ already presented in \cite{RR08}, while the extensions needed to conform the new $\sqlp{\simrel}{\qdom}$ scheme are presented in subsection \ref{SQLP}.

\subsection{Qualified Logic Programming} \label{QLP}

$\qlp{\qdom}$ was proposed in our previous work \cite{RR08} as a generic scheme for qualified logic programming over a given qualification domain $\qdom$. In that scheme, a \emph{signature} $\Sigma$ providing constructor and predicate symbols with given arities is assumed. \emph{Terms} are built from constructors and \emph{variables} from a countably infinite set $\Var$ (disjoint from $\Sigma$)  and \emph{Atoms} are of the form $p(t_1, \ldots, t_n)$ (shortened as $p(\ats{t_n})$ or simply $p(\ats{t})$) where $p$ is a $n$-ary predicate symbol and $t_i$ are terms. We write $\Atz$, called the \emph{open Herbrand base}, for the set of all atoms. A $\qlp{\qdom}$ program $\Prog$ is a finite set of \emph{$\qdom$-qualified definite Horn clauses} of the form $A \qgets{d} \ats{B}$ where $A$ is an atom, $\ats{B}$ a finite conjunction of atoms and $d \in \dqdom$ is the \emph{attenuation value} attached to the clause's implication.

As explained in  \cite{RR08}, in our aim to work with qualifications we are not only interested in just proving an atom, but in proving it along with a qualification value. For this reason, \emph{$\qdom$-qualified atoms} ($\at{A}{d}$ where $A$ is an atom and $d \in \dqdom$) are introduced to represent the statement that the atom $A$ holds for \emph{at least} the qualification value $d$. For use in goals to be solved, \emph{open $\qdom$-annotated atoms} ($\at{A}{W}$ where $A$ is an atom and $W$ a \emph{qualification variable} intended to take values over $\qdom$) are also introduced, and a countably infinite set $\War$ of qualification variables (disjoint from $\Var$ and $\Sigma$) is postulated. The \emph{annotated Herbrand base} over $\qdom$ is defined as the set $\QAtz$ of all $\qdom$-qualified atoms. A \emph{$\qdom$-entailment relation} over $\QAtz$, defined as $\at{A}{d} \Dentail \at{A'}{d'}$ iff there is some substitution $\theta$ such that $A' = A\theta$ and $d' \sqsubseteq d$, is used to formally define an \emph{open Herbrand interpretation} over $\qdom$ --from now on just an \emph{interpretation}-- as any subset $\I \subseteq \QAtz$ which is closed under $\qdom$-entailment. We write $\intd$ for the family of all interpretations. The notion of model is such that given any clause $C \equiv A \qgets{d} B_1, \ldots, B_k$ in the $\qlp{\qdom}$ program $\Prog$, an interpretation $\I$ is said to be a \emph{model} of $C$ iff for any substitution $\theta$ and any qualification values $d_1, \ldots, d_k \in \dqdom$ such that $\at{B_i\theta}{d_i} \in \I$ for all $1 \le i \le k$, one has $\at{A\theta}{(d \circ \infi \{d_1, \ldots, d_k\})} \in \I$. The interpretation $\I$ is also said to be a model of the $\qlp{\qdom}$ program $\Prog$ (written as $\I \models \Prog$) iff it happen to be a model of every clause in $\Prog$.

As technique to infer formulas (or in our case $\qdom$-qualified atoms) from a given $\qlp{\qdom}$ program $\Prog$, and following traditional ideas, we consider two alternative ways of formalizing an inference step which goes from the body of a clause to its head: both an interpretation transformer $\Tp : \intd \to \intd$, and a qualified variant of Horn Logic, noted as $\qhl{\qdom}$, called \emph{Qualified Horn Logic} over $\qdom$. As both methods are equivalent and correctly characterize the least Herbrand model of a given program $\Prog$, we will only be recalling the logic $\qhl{\qdom}$, although we encourage the reader to see Section 3.2 in \cite{RR08}, where the fix-point semantics is explained.

The logic $\qhl{\qdom}$ is defined as a deductive system consisting just of one inference rule: $\mbox{QMP}(\qdom)$, called \emph{Qualified Modus Ponens} over $\qdom$. Such rule allows us to give the following inference step given that there were some $(A \qgets{d} B_1, \ldots, B_k) \in \Prog$, some substitution $\theta$ such that $A' = A\theta$ and $B'_i = B_i\theta$ for all $1 \le i \le k$ and some $d' \in \dqdom$ such that $d' \sqsubseteq d \circ \infi \{d_1, \ldots, d_k\}$:
\[
    \frac
    {\quad \at{B'_{1}}{d_{1}} \quad \cdots \quad \at{B'_{k}}{d_{k}} \quad}
    {\at{A'}{d'}} \quad \mbox{QMP}(\qdom)
\]
Roughly, each $\mbox{QMP}(\qdom)$ inference step using an instance of a program clause $A \qgets{d} \ats{B}$ has the effect of propagating to the head the {\em qualification value} $d \circ b$, where $b$ is the infimum in $\qdom$ of the qualification values  $d_i \in \dqdom$ previously computed for the various atoms occurring in the body. This helps to understand the claims made in Example \ref{someDomains} above about the intended use of elements of the domains $\U$ and $\W$ for qualifying logical assertions. We use the notations $\Prog \qhld \at{A}{d}$ (resp. $\Prog \qhldn{n} \at{A}{d}$) to indicate that $\at{A}{d}$ can be inferred from the clauses in program $\Prog$ in finitely many steps (resp. $n$ steps). The \emph{least Herbrand model of $\Prog$} happens to be $\Mp = \{\at{A}{d} \mid \Prog \qhld \at{A}{d}\}$, as proved in \cite{RR08}.

\subsection{Similarity-based Qualified Logic Programming} \label{SQLP}

\begin{figure}
\begin{center}
\tt
\begin{tabular}{|r@{\hspace{0.2cm}}l|}
\hline
&\\
\scriptsize 1 & wild(lynx) <-0.9- \\
\scriptsize 2 & wild(boar) <-0.9- \\
\scriptsize 3 & wild(snake) <-1.0- \\
&\\
\scriptsize 4 & farm(cow) <-1.0- \\
\scriptsize 5 & farm(pig) <-1.0- \\
&\\
\scriptsize 6 & domestic(cat) <-0.8- \\
\scriptsize 7 & domestic(snake) <-0.4- \\
&\\
\scriptsize 8 & intelligent(A) <-0.9- domestic(A) \\
\scriptsize 9 & intelligent(lynx) <-0.7- \\
&\\
\scriptsize 10 & pacific(A) <-0.9- domestic(A) \\
\scriptsize 11 & pacific(A) <-0.7- farm(A) \\
&\\
\scriptsize 12 & pet(A) <-1.0- pacific(A), intelligent(A) \\
&\\
\hline
&\\
&$\simrel$(farm,domestic) = 0.3 \\
&$\simrel$(pig,boar) = 0.7 \\
&$\simrel$(lynx,cat) = 0.8 \\
&\\
\hline
\end{tabular}
\end{center}
\caption{$\sqlp{\simrel}{\U}$ program.\label{fig:example}}
\end{figure}

The scheme $\sqlp{\simrel}{\qdom}$ presented in this subsection has two parameters $\simrel$ and $\qdom$, where $\qdom$ can be any qualification domain and $\simrel$ can be any admissible $\qdom$-valued similarity relation, in the sense of Definition \ref{defSR}. The new scheme subsumes the approach in \cite{RR08} by behaving as $\qlp{\qdom}$ in the case that $\simrel$  is chosen as the identity, and it also subsumes similarity-based $LP$ by behaving as the approach in \cite{Ses02} and related papers in  the case that $\qdom$ is chosen as $\U$.

Syntactically, $\sqlp{\simrel}{\qdom}$ presents almost no changes w.r.t. $\qlp{\qdom}$, but the declarative semantics must be extended to account for the behavior of the parametrically given similarity relation $\simrel$.  As in the previous subsection, we assume a signature $\Sigma$ providing again constructor and predicate symbols. \emph{Terms} and \emph{Atoms} are built the same way they were in $\qlp{\qdom}$, and $\Atz$ will stand again for the set of all atoms, called the \emph{open Herbrand base}. An atom $A$ is called {\em linear} if there is no variable with multiple occurrences in $A$; otherwise $A$ is called {\em non-linear}. A $\sqlp{\simrel}{\qdom}$ program $\Prog$ is a finite set of \emph{$\qdom$-qualified definite Horn clauses} with the same syntax as in $\qlp{\qdom}$, along with a $\qdom$-valued admissible similarity relation $\simrel$ in the sense of Definition \ref{defSR}, item 2. Figure \ref{fig:example} shows a simple $\sqlp{\simrel}{\U}$ program built from the similarity relation $\simrel$ given in the same figure and the qualification domain $\U$ for certainty values. This program will be used just for illustrative purposes in the rest of the paper. The reader is referred to Section \ref{Domains} for other examples of qualification domains, and to the references \cite{LSS04,MOV04} for suggestions concerning practical applications of similarity-based $LP$.

\emph{$\qdom$-qualified atoms} ($\at{A}{d}$ with $A$ an atom and $d \in \dqdom$) and \emph{open $\qdom$-annotated atoms} ($\at{A}{W}$ with $A$ and atom and $W \in War$ a qualification variable intended to take values in $\dqdom$) will still be used here. Similarly, the \emph{annotated open Herbrand base} over $\qdom$ is again defined as the set $\QAtz$ of all $\qdom$-qualified atoms. At this point, and before extending the notions of $\qdom$-entailment relation and interpretation to the $\sqlp{\simrel}{\qdom}$ scheme, we need to define what an $\simrel$-instance of an atom is.  Intuitively, when building $\simrel$-instances of an atom $A$, signature symbols occurring in $A$ can be replaced by similar ones, and different occurrences of the same variable in $A$ may be replaced by different terms, whose degree of similarity must be taken into account. Technically, $\simrel$-instances of an atom $A \in \Atz$ are built from a linearized version of $A$ which has the form $\mbox{lin}(A) = (\linear{A}, \linear{\Set})$ and is constructed as follows: $\linear{A}$ is a linear atom built from $A$ by replacing each $n$ additional occurrences of a variable $X$ by new fresh variables $X_i$ $(1 \leq i \leq n)$; and $\linear{\Set}$ is a set of \emph{similarity conditions} $X \sim X_i$ (with $1 \le i \le n$) asserting the similarity of all variables in $\linear{A}$ that correspond to the same variable $X$ in $A$. As a concrete illustration, let us show the linearization of two atoms. Note what happens when the atom $A$ is already linear as in the first case: $\linear{A}$ is just the same as $A$ and $\linear{\Set}$ is empty.

    \begin{itemize}
        \item $H_1 = p(c(X),Y)$ \\
        $\mbox{lin}(H_1) = (p(c(X),Y), \, \{\})$

        \item $H_2 = p(c(X),X,Y)$ \\
        $\mbox{lin}(H_2) = (p(c(X),X_1,Y), \, \{X \sim X_1\})$
    \end{itemize}

Now we are set to formally define the $\simrel$-instances of an atom.

\begin{definition}\label{def:r-instance-atom}
($\simrel$-instance of an atom). Assume an atom $A \in \Atz$ and its linearized version $\mbox{lin}(A) = (\linear{A}, \linear{\Set})$. Then, an atom $A'$ is said to be an $\simrel$-instance of $A$ with similarity degree $\delta$, noted as $(A',\delta) \in [A]_\simrel$, iff there are some atom $A^\Set$ and some substitution $\theta$ such that $A' = A^\Set\theta$ and $\delta = \simrel(\linear{A}, A^\Set) ~\sqcap~ \infi\{\simrel(X_i\theta, X_j\theta) \mid (X_i \sim X_j) \in \linear{\Set}\} \neq \bot$.
\end{definition}

Next, the \emph{$(\simrel, \qdom)$-entailment relation} over $\QAtz$ is defined as follows: $\at{A}{d} \DentailSim \at{A'}{d'}$ iff there is some similarity degree $\delta$ such that $(A',\delta) \in [A]_\simrel$ and $d' \sqsubseteq d \circ \delta$. Finally, an \emph{open Herbrand interpretation} --just \emph{interpretation} from now on-- over $(\simrel,\qdom)$ is defined as any subset $\I \in \QAtz$ which is closed under $(\simrel, \qdom)$-entailment. That is, an interpretation $\I$ including a given $\qdom$-qualified atom $\at{A}{d}$ is required to include all the `similar instances' $\at{A'}{d'}$ such that $\at{A}{d} \DentailSim \at{A'}{d'}$, because we intend to formalize a semantics in which all such similar instances are valid whenever $\at{A}{d}$ is valid. This complements the intuition given for the $\qdom$-entailment relation in $\qlp{\qdom}$ to include the similar instances (obtainable due to $\simrel$) of each atom, and not only those which are true because we can prove them for a better (i.e. higher in $\qdom$) qualification. Note that $(\simrel,\qdom)$-entailment is a refinement of $\qdom$-entailment, since: $\at{A}{d} \Dentail \at{A'}{d'}$ $\Longrightarrow$ there is some substitution $\theta$ such that $A' = A\theta$ and $d' \sqsubseteq d$ $\Longrightarrow$ $(A',\top) \in [A]_\simrel$ and $d' \sqsubseteq d \circ \top$ $\Longrightarrow$ $\at{A}{d} \DentailSim \at{A'}{d'}$.

As an example of the closure of interpretations w.r.t. $(\simrel, \qdom)$-entailment, consider the $\U$-qualified atom \texttt{domestic(cat)\#0.8}. As a trivial consequence of Proposition \ref{prop:least-model} below,  this atom belongs to the least Herbrand model of the program in Figure \ref{fig:example}. On the other hand, we also know that  \texttt{lynx} is similar to  \texttt{cat} with a similarity degree of $0.8$ w.r.t. the similarity relation $\simrel$ in Figure \ref{fig:example}. Therefore, \texttt{domestic(lynx)} is a $\simrel$-instance of  \texttt{domestic(cat)} to the degree $0.8$. Then, by definition of $(\simrel, \U)$-entailment, it turns out that \texttt{domestic(cat)\#0.8} $\succcurlyeq_{(\simrel,\U)}$ \texttt{domestic(lynx)\#0.64}, and the $\U$-qualified atom \texttt{domestic(lynx)\#0.64} does also belong to the least model of the example program. Intuitively, $0.64 = 0.8 \times 0.8$ is the best $\U$-qualification which can be inferred from the $\U$-qualification $0.8$ for  \texttt{domestic(cat)} and the $\simrel$-similarity $0.8$ between  \texttt{domestic(cat)} and  \texttt{domestic(lynx)}.

We will write $\intrd$ for the family of all interpretations over $(\simrel,\qdom)$, a family for which the following proposition can be easily proved from the definition of an interpretation and the definitions of the union and intersection of a family of sets.

\begin{proposition}\label{prop:lattice}
The family $\intrd$ of all interpretations over $(\simrel,\qdom)$ is a complete lattice under the inclusion ordering $\subseteq$, whose extreme points are $\intrd$ as maximum and $\emptyset$ as minimum. Moreover, given any family of interpretations $I \subseteq \intrd$, its lub and glb are $\infi I = \union \{\I \in \intrd \mid \I \in I\}$ and $\supr I = \inter \{\I \in \intrd \mid \I \in I\}$, respectively.
\end{proposition}

Similarly as we did for the $\simrel$-instances of an atom, we will define what the $\simrel$-instances of a clause are. The following definition tells us so.

\begin{definition}\label{def:r-instance-clause}
($\simrel$-instance of a clause). Assume a clause $C \equiv A \qgets{d} B_1, \ldots, B_k$ and the linearized version of its head atom $\mbox{lin}(A) = (\linear{A}, \linear{\Set})$. Then, a clause $C'$ is said to be an $\simrel$-instance of $C$ with similarity degree $\delta$, noted as $(C',\delta) \in [C]_\simrel$, iff there are some atom $A^\Set$ and some substitution $\theta$ such that $\delta = \simrel(\linear{A}, A^\Set) ~\sqcap~ \infi\{\simrel(X_i\theta,$ $X_j\theta) \mid (X_i \sim X_j) \in \linear{\Set}\} \neq \bot$ and $C' \equiv A^\Set\theta \qgets{d} B_1\theta, \ldots, B_k\theta$.
\end{definition}

Note that as an immediate consequence from Definitions \ref{def:r-instance-atom} and \ref{def:r-instance-clause} it is true that given two clauses $C$ and $C'$ such that $(C',\delta) \in [C]_\simrel$, and assuming $A$ to be head atom of $C$ and $A'$ to be the head atom of $C'$, then we have that $(A',\delta) \in [A]_\simrel$.

Let $C$ be any clause $A \qgets{d} B_1, \ldots, B_k$ in the program $\Prog$, and $\I \in \intrd$ any interpretation over $(\simrel,\qdom)$. We say that $\I$ is a model of $C$ iff for any clause $C' \equiv H' \qgets{d} B'_1, \ldots, B'_k$ such that $(C',\delta) \in [C]_\simrel$ and any qualification values $d_1, \ldots, d_k \in \dqdom$ such that $\at{B'_i}{d_i} \in \I$ for all $1 \le i \le k$, one has $\at{H'}{d'} \in \I$ where $d' = d \circ \infi \{e, d_1, \ldots, d_k\}$. And we say that $\I$ is a model of the $\sqlp{\simrel}{\qdom}$ program $\Prog$ (also written $\I \models \Prog$) iff $\I$ is a model of each clause in $\Prog$.


We will provide now a way to perform an inference step from the body of a clause to its head. As in the case of  $\qlp{\qdom}$, this can be formalized in two alternative ways, namely an interpretation transformer and a variant of Horn Logic. Both approaches lead to equivalent characterizations of least program models. Here we focus on the second approach, defining what we will call \emph{Similarity-based Qualified Horn Logic} over $(\simrel,\qdom)$ --abbreviated as $SQHL(\simrel,\qdom)$--, another variant of Horn Logic and an extension of the previous $QHL(\qdom)$. The logic $SQHL(\simrel,\qdom)$ is also defined as a deductive system consisting just of one inference rule $SQMP(\simrel,\qdom)$, called \emph{Similarity-based Qualified Modus Ponens} over $(\simrel,\qdom)$:

If $((A' \qgets{d} B'_1, \ldots, B'_k), \delta) \in [C]_\simrel$ for some clause $C \in \Prog$ with attenuation value $d$, then the following inference step is allowed for any $d' \in \dqdom$ such that $d' \sqsubseteq d \circ \infi \{\delta, d_1, \ldots, d_k\}$:
\[
    \frac
    {\quad \at{B'_{1}}{d_{1}} \quad \cdots \quad \at{B'_{k}}{d_{k}} \quad}
    {\at{A'}{d'}} \quad \mbox{SQMP}(\simrel,\qdom) \enspace .
\]

We will use the notations $\Prog \sqhlrd \at{A}{d}$ (respectively $\Prog \sqhlrdn{n} \at{A}{d}$) to indicate that $\at{A}{d}$ can be inferred from the clauses in program $\Prog$ in finitely many steps (respectively $n$ steps). Note that $\sqhl{\simrel}{\qdom}$ proofs can be naturally represented as upwards growing \emph{proof trees} with $\qdom$-qualified atoms at their nodes, each node corresponding to one inference step having the children nodes as premises.

The following proposition contains the main result concerning the declarative semantics of the $\sqlp{\simrel}{\qdom}$ scheme. A full proof can be developed in analogy to the $\qlp{\qdom}$ case presented in \cite{RR08,RR08TR}.

\begin{proposition}\label{prop:least-model}
Given any $\sqlp{\simrel}{\qdom}$ program $\Prog$. The \emph{least Herbrand model} ($\Mp$) of $\Prog$ is \[\{\at{A}{d} \mid \Prog \sqhlrd \at{A}{d}\} \enspace .\]
\end{proposition}

The following example serves as an illustration of how the logic $\sqhl{\simrel}{\qdom}$ works over $(\simrel,\U)$ using the example program displayed in Figure  \ref{fig:example}.

\begin{example}\label{ex:least-model}
The following proof tree proves that the atom \emph{\texttt{pet(ly\-nx)}} can be inferred for at least a qualification value of $0.50$ in the $\sqlp{\simrel}{\U}$ program $\Prog$ of Figure \ref{fig:example}. Let's see it:

$$
\displaystyle\frac
{
  \,\displaystyle\frac
  {
    \,\displaystyle\frac
    {}
    {\mbox{\tt\small domestic(lynx)\#0.64}} \mbox{\tiny (4)}
  }
  {\mbox{\tt\small pacific(lynx)\#0.57}} \mbox{\tiny (2)}
  \,\,
  \displaystyle\frac
  {}
  {\mbox{\tt\small intelligent(lynx)\#0.70}} \mbox{\tiny (3)}
}
{\mbox{\tt\small pet(lynx)\#0.50}} \mbox{\tiny (1)}
$$

\noindent where the clauses and qualification values used for each inference step are:

\begin{enumerate}
    \item[(1)] \emph{\texttt{pet(lynx) <-1.0- pacific(lynx),intelligent(lynx)}} is an instance of clause $12$ in $\Prog$ and $0.50 \le 1.0 \times \mathrm{min} \{1.0,$ $0.57,$ $0.70\}$. Note that the first $1.0$ in the minimum is the one which comes from the similarity relation as for this step we are just using a plain instance of clause $12$ in $\Prog$.
    \item[(2)] \emph{\texttt{pacific(lynx) <-0.9- domestic(lynx)}} is a plain instance of clause $10$ in $\Prog$ and $0.57 \le 0.9 \times \mathrm{min} \{1.0, 0.64\}$.
    \item[(3)] \emph{\texttt{intelligent(lynx) <-0.7-}} is clause $9$ in $\Prog$ and $0.70 \le 0.70 \times \mathrm{min} \{1.0\}$.
    \item[(4)] The clause \emph{\texttt{domestic(lynx) <-0.8-}} is an $\simrel$-instance of clause $6$ with a similarity degree of $0.8$ and we have $0.64 \le 0.8 \times \mathrm{min} \{0.8\}$. \qed
\end{enumerate}
\end{example}

\section{Reducing Similarities to Qualifications} \label{Reduction}

\subsection{A Program Transformation} \label{sec:PT}

In this section we prove that any $\sqlp{\simrel}{\qdom}$ program $\Prog$ can be transformed into an equivalent $\qlp{\qdom}$ program which will be denoted by $\trans{\simrel}{\Prog}$. The program transformation is defined as follows:

\begin{definition} \label{def:trans}
Let $\Prog$ be a $\sqlp{\simrel}{\qdom}$ program. We define the transformed program $\trans{\simrel}{\Prog}$ as:
$$
\trans{\simrel}{\Prog} = \Prog_S \cup \Prog_\sim \cup
\Prog_{\mathrm{pay}}
$$
where the auxiliary sets of clauses $\Prog_S$, $\Prog_\sim$, $\Prog_{\mathrm{pay}}$ are defined as:
\begin{itemize}
\item For each clause $(H \qgets{d} \ats{B}) \in \Prog$ and for each $H'$ such that $\simrel(\linear{H},H') \neq \bot$
$$
(H' \qgets{d} pay_{\simrel(\linear{H},H')}, \linear{S}, \ats{B}) \in
\mathcal{P}_S
$$
where $(\linear{H}, \linear{S}) = lin(H)$.
\item $\Prog_\sim = \{X \sim X \qgets{\top} \}$ $\cup$
      $\{ (c(\overline{X}_n) \sim c'(\overline{Y}_n) \qgets{\top} pay_{\simrel(c,c')}, X_1 \sim Y_1, \dots, X_n \sim Y_n)  \mid c, c' \in CS$ of arity $n$, $\simrel(c,c') \neq \bot \}$
\item $\Prog_{\mathrm{pay}} = \{ (pay_w \qgets{w}) \mid $ for each atom $pay_w$ occurring in $\Prog_\sim \cup \Prog_S \}$
\end{itemize}
\end{definition}

Note that the linearization of clause heads in this transformation is motivated by the role of linearized atoms in the $SQHL(\simrel,\qdom)$ logic defined in Subsection \ref{SQLP} to specify the declarative semantics of $\sqlp{\simrel}{\qdom}$ programs. For instance, assume  a $\sqlp{\simrel}{\U}$ program $\Prog$ including the clause $p(X,X) \qgets{1.0}$ and two nullary constructors $c$, $d$ such that $\simrel(c,d) = 0.8$. Then, $SQHL(\simrel,\U)$ supports the derivation $\Prog \sqhlx{\simrel}{\U}{} \at{p(c,d)}{0.8}$, and the transformed program $\trans{\simrel}{\Prog}$ will include the clauses
$$
\begin{array}{lll}
            & p(X, X_1) & \qgets{1.0} pay_{1.0}, X \sim X_1,\\
                            & X \sim X     & \qgets{1.0},\\
                            & c \sim d     & \qgets{1.0} pay_{0.8},\\
                            & pay_{1.0}    & \qgets{1.0}, \\
                            & pay_{0.8}    & \qgets{0.8}
\end{array}
$$
thus enabling the corresponding derivation $\trans{\simrel}{\Prog} \qhlx{\U}{} \at{p(c,d)}{0.8}$ in $QHL(\U)$.

%

In general, $\Prog$ and $\trans{\simrel}{\Prog}$ are semantically equivalent in the sense that  $\Prog \sqhlrd \at{A}{d}\ \Longleftrightarrow\ \trans{\simrel}{\Prog} \qhld \at{A}{d}$ holds for any $\qdom$-qualified atom $\at{A}{d}$, as stated in Theorem \ref{th:equivalence} below. The next technical lemma will be useful for the proof of this theorem.

\begin{lemma} \label{lema:equiv}
Let $\mathcal{P}$ be a $\sqlp{\simrel}{\qdom}$ program and $\trans{\simrel}{\Prog}$ its transformed program according to Definition \ref{def:trans}. Let $t,s$ be two terms in $\Prog$'s signature and $d \in \dqdom$. Then:
\begin{enumerate}
    \item \label{lema:equiv:1} $\trans{\simrel}{\Prog} \qhld \at{(t \sim s)}{d} \Longrightarrow d \sqsubseteq  \simrel(t,s)$
    \item \label{lema:equiv:2} $\simrel(t,s) = d \Longrightarrow \trans{\simrel}{\Prog} \qhld \at{(t \sim s)}{d}$
\end{enumerate}
\end{lemma}

\begin{proof} We prove the two items separately.

\begin{enumerate}
    \item Let $T$ be a $\qhl{\qdom}$ proof tree witnessing $$\trans{\simrel}{\Prog} \qhld \at{(t \sim s)}{d}$$ We prove by induction on number of nodes of $T$ that $d \sqsubseteq  \simrel(t,s)$. The basis case, with $T$ consisting of just one node, must correspond to some inference without premises, i.e., a clause with empty body for $\sim$. Checking $P_\sim$ we observe that $X \sim X \qgets{\top}$ is the only possibility. In this case $t$ and $s$ must be the same term and by the reflexivity of $\simrel$ (Def. \ref{defSR}), $\simrel(t,s)= \top$, which means $d \sqsubseteq \simrel(t,s)$ for every $d$. In the inductive step, we consider $T$ with more than one node. Then the inference step at the root of $T$ uses some clause $(c(\overline{X}_n) \sim c'(\overline{X}'_n) \qgets{\top}pay_{\simrel(c,c')}, X_1 \sim X'_1, \dots, X_n \sim X'_n) \in \Prog_\sim$, and must be of the form:
        \[ \frac
            {\quad \at{pay_w}{v}\ \ \at{(t_1 \sim s_1)}{e_1}\ \dots \ \at{(t_n \sim s_n)}{e_n} \quad}
            {\at{c(\overline{t}_n) \sim c'(\overline{s}_n)}{d}}
        \]
        where $w=\simrel(c,c')$, $v\in \qdom$, $v \sqsubseteq w$, $t = c(\overline{t}_n)$, $s=c'(\overline{s}_n)$, and $e_1, \dots, e_n$ s.t. $d \sqsubseteq \top \circ \infi \{v, e_1, \ldots, e_k\}$, i.e., $d \sqsubseteq \infi\{v, e_1, \ldots, e_k\}$. By induction hypothesis $e_i \sqsubseteq \simrel(t_i,s_i)$ for $i=1 \dots n$. Then $d \sqsubseteq \infi \{v, e_1, \ldots, e_n\}$ implies $d \sqsubseteq$ $\infi \{w, \simrel(t_1,s_1),$ $\ldots, \simrel(t_n,s_n)\}$ and hence $d \sqsubseteq \simrel(t,s)$ (Def. \ref{def:ER}, item 3).

    \item If $\simrel(t,s) = d$, $d \neq \bot$, we prove that $\trans{\simrel}{\Prog} \qhld \at{(t \sim s)}{d}$ by  induction on the syntactic structure of $t$. The basis corresponds to the case $t = c$ for some constant $c$, or $t=Y$ for some variable $Y$. If $t=c$ then $s=c'$ for some other constant $c'$. By Definition \ref{def:trans} there is a clause in $\Prog_\sim$ of the form $(c \sim c' \qgets{\top} pay_{d})$. Using this clause and the identity substitution we can write the root inference step of a proof for $\trans{\simrel}{\Prog} \qhld\at{(c \sim c')}{d}$ as follows:
        \[ \frac
            {\at{pay_d}{d}}
            {\quad \at{c \sim c'}{d} \quad}
        \]
        The condition required by the inference rule $\mbox{QMP}(\qdom)$ is in this particular case $d \sqsubseteq \top \circ \infi \{ d \}$, and $\top \circ \infi \{ d \}=d$. Proving the only premise $\at{pay_d}{d}$  in $\qhl{\qdom}$ is direct from its definition. If $t=Y$, with $Y$ a variable, then $s=Y$ and $d=\top$ (otherwise $\simrel(t,s) = \bot$). Then $\trans{\simrel}{\Prog} \qhld \at{(Y \sim Y)}{\top}$ can be proved by using the clause $(X \sim X \qgets{\!\!\!\top}) \in \Prog_\sim$ with substitution $\theta=\{ X \mapsto Y \}$.

In the inductive step, $t$ must be of the form $c(\overline{t}_n)$, with $n \geq 1$, and then $s$ must be of the form $c'(\overline{s}_n)$ (otherwise $\simrel(t,s) = \bot$). From $d=\simrel(t,s) \neq \bot$ (hypotheses of the lemma) and Definition \ref{def:ER} we have that $\simrel(c,c') \neq \bot$. Then, by Definition \ref{def:trans}, there is a clause in $ \Prog_\sim$ of the form: {\small $$c(\overline{X}_n) \sim c'(\overline{Y}_n) \qgets{\top} pay_{\simrel(c,c')}, X_1 \sim Y_1, \dots, X_n \sim Y_n$$ } By using the substitution $\theta=\{X_1 \mapsto t_1, \dots, X_n \mapsto t_n, Y_1 \mapsto s_1, \dots, Y_n \mapsto s_n \}$ we can write the root inference step in $\qhl{\qdom}$ as:
    \[ \frac
        {\quad \at{pay_{\simrel(c,c')}}{\simrel(c,c')}\ (\at{t_i \sim s_i}{\simrel(t_i,s_i)})_{i=1\dots n} \quad}
        {\at{c (\overline{t}_n) \sim c'(\overline{s}_n)}{d}}
    \]

The inference can be applied because the condition $$d \sqsubseteq \top \circ \infi \{\simrel(c,c'), \simrel(t_1,s_1), \dots, \simrel(t_n,s_n) \}$$ reduces to $$d \sqsubseteq \infi \{\simrel(c,c'),  \simrel(t_1,s_1), \dots, \simrel(t_n,s_n) \}$$ which holds by Definition \ref{def:ER}, item 3. Moreover, the premises $\at{t_i \sim s_i}{\simrel(t_i,s_i)}$, $i=1 \dots n$, hold in $\qhl{\qdom}$ due to the inductive hypotheses, and proving $${\at{pay_{\simrel(c,c')}}{\simrel(c,c')}}$$ is straightforward from its definition.\qedhere
\end{enumerate}
\end{proof}

Now we can prove the equivalence between semantic inferences in $\qhl{\qdom}$ w.r.t. $\mathcal{P}$ and semantic inferences in $\sqhl{\simrel}{\qdom}$ w.r.t. $\trans{\simrel}{\Prog}$.

\begin{theorem} \label{th:equivalence}
Let $\mathcal{P}$ be a $\sqlp{\simrel}{\qdom}$ program,  $A$ an atom in $\Prog$'s signature and $d \in \dqdom$. Then:
$$
\Prog \sqhlrd \at{A}{d}\ \Longleftrightarrow\ \trans{\simrel}{\Prog}
\qhld \at{A}{d} \enspace .
$$
\end{theorem}

\begin{proof}
Let $T$ be a $\sqhl{\simrel}{\qdom}$ proof tree for some annotated atom $\at{A}{d}$ in $\Prog$'s signature witnessing $\Prog \sqhlrd \at{A}{d}$. We prove that $\trans{\simrel}{\Prog} \qhld \at{A}{d}$ by induction on the number of nodes of $T$.

The inference step at the root of $T$  must be of the form
\[
    \frac
    {\quad \at{B'_{1}}{d_{1}} \quad \cdots \quad \at{B'_{k}}{d_{k}} \quad}
    {\at{A}{d}} \quad \mbox{(1)} \enspace
\]
with $((A \qgets{e} B'_1, \ldots, B'_k), \delta) \in [C]_\simrel$ for some clause $C \equiv (H \qgets{e} B_1, \ldots, B_k)  \in \Prog$ (observe that the case $k=0$ corresponds to the induction basis). By Definition \ref{def:r-instance-clause}, $A = H'\theta$, $B'_i= B_i\theta$ for some substitution $\theta$ and atom $H'$ such that $\delta = \simrel(\linear{H}, H') \sqcap \infi\{\simrel(X_i\theta,$ $X_j\theta) \mid (X_i \sim X_j) \in \linear{\Set}\} \neq \bot$, with $\mbox{lin}(H) = (\linear{H}, \linear{\Set})$. This means in particular that $w = \simrel(\linear{H}, H') \neq \bot$, which by Definition \ref{def:trans} implies that  there is a clause $C'$ in $\trans{\simrel}{\Prog}$ of the form $C'\equiv (H' \qgets{e} pay_w, \linear{S}, B_1, \dots B_k)$. Then the root inference step of the deduction proving ${\Prog} \qhld \at{A}{d}$ will use the inference rule $\mbox{QMP}(\qdom)$ with $C'$ and substitution $\theta$ (such that $H'\theta = A$) as follows:
\[ \frac
    {\at{pay_w\theta}{w}\quad (\at{(u_i \sim v_i)\theta}{e_i})_{1\leq i \leq m}\quad \at{B'_1}{d_1} \cdots \at{B'_k}{d_k}}
    {\at{A}{d} }\quad \mbox{(2)}
\]
where $\linear{S} = \{ u_1 \sim v_1, \dots, u_m \sim v_m \}$, and $e_i = \simrel(u_i\theta, v_i\theta)$ for $i=1 \dots m$.

Next we check that the premises can be proved from $\trans{\simrel}{\Prog}$ in $\qhl{\qdom}$:
\begin{itemize}
    \item $pay_w\theta = pay_w$, since $pay_w$ is a nullary predicate for every $w$. Therefore $\trans{\simrel}{\Prog} \qhld \at{pay_w}{w}$ is immediate from the definition of $pay_w$ in Definition \ref{def:trans}.

    \item For each $1 \le i \le m$, we observe that $\simrel(u_i\theta, v_i\theta) \neq \bot$ because $\delta \neq \bot$ has been computed above as the infimum of a set including $\simrel(u_i\theta, v_i\theta)$ among its members. Then $\trans{\simrel}{\Prog} \qhld (u_i \sim v_i)\theta$ holds by Lemma \ref{lema:equiv}, item \ref{lema:equiv:2}.

    \item For each $1 \le i \le k$, (1) shows that $\Prog \sqhlrd \at{B'_i}{d_i}$ with a proof tree having less nodes that $T$. Therefore, $\trans{\simrel}{\Prog} \qhld \at{B'_i}{d_i}$ by induction hypothesis.
\end{itemize}

In order to perform the inference step (2), the QMP($\qdom$) inference rule also requires that $d \sqsubseteq e \circ \infi \{w, e_1 \dots, e_m, \, d_1, \ldots, d_k\}$. This follows from the associativity of $\sqcap$ since:
\begin{itemize}
    \item As defined above, $\delta = \simrel(\linear{H}, H')$ $\sqcap~ \infi\{\simrel(X_i\theta, X_j\theta) \mid (X_i \sim X_j) \in \linear{\Set}\}$, i.e. $\delta = w$ $\sqcap~ \infi\{e_1 \dots e_m \}$.
    \item By the $\mbox{SQMP}(\simrel,\qdom)$  inference (1) we know that $d \sqsubseteq e \circ \infi \{\delta, d_1, \ldots, d_k\}$.
\end{itemize}
\medskip

Let $T$ be a $\qhl{\qdom}$ proof tree witnessing $\trans{\simrel}{\Prog}$ $\qhld$ $\at{A}{d}$ for some atom $A$ in $\Prog$'s signature.  We prove by induction on the number of nodes of $T$ that ${\Prog} \sqhlrd \at{A}{d}$.

Since $A$ is in $\Prog$'s signature, the clause employed at the inference step at the root of $T$ must be in the set $\Prog_S$ of Definition \ref{def:trans}, and the inference step at the root of $T$ have of the form of the inference (2) above. Hence this clause must have been constructed from a clause $C \equiv (H \qgets{e} B_1, \ldots, B_k) \in \Prog$ and some atom $H'$ such that $A = H'\theta$ and $\simrel(\linear{H},H') \neq \bot$, where $lin(H)=(\linear{H}, \linear{S})$.

Then we can use $C$ and $\theta$ to prove ${\Prog} \sqhlrd \at{A}{d}$ by  a $\mbox{SQMP}(\simrel,\qdom)$ inference like (1) using the $\simrel$-instance $C' \equiv A \qgets{e} B'_1, \dots, B'_k$ of $C$. The premises can be proved in $\sqhl{\simrel}{\qdom}$ by induction hypotheses, since all of them are also premises in (2). Finally, we must check that the conditions required by (1) hold:  $(C', \delta) \in [C]_\simrel$ for some $\delta \in \qdom$, $\delta \neq \bot$ s.t. $d \sqsubseteq e \circ \infi \{\delta, d_1, \ldots, d_k\}$. This is true for $\delta = \infi \{w, e'_1, \dots, e'_n\}$, with $e'_i=\simrel(u_i\theta,v_i\theta)$ for $i=1 \dots m$. Observe that in the premises of (2) we have $\qhl{\qdom}$ proofs of $\at{u_i\theta \sim v_i\theta}{e_i}$ for $i=1 \dots m$. Therefore $e_i \sqsubseteq e'_i$, by Lemma \ref{lema:equiv}, item \ref{lema:equiv:1}. Then
$$
\begin{array}{lll}
d & \sqsubseteq e \circ \infi \{w, e_1 \dots, e_m, \, d_1, \ldots,
d_k\} & \mathrm{(by (2))} \\
  & \sqsubseteq e \circ \infi \{w, e'_1 \dots, e'_m, \, d_1, \ldots,
d_k\} & (e_i \sqsubseteq e'_i)\\
  & = e \circ \infi \{\delta,d_1, \ldots, d_k\}
\end{array}
$$

We must still prove that $\delta \neq \bot$. Observe that by the distributivity of $\circ$ w.r.t. $\sqcap$ (Def. \ref{defQD}, axiom 2.(e)): $$e \circ \infi \{\delta, d_1, \ldots, d_k\} = (e \circ \delta) \sqcap (e \circ \infi \{d_1, \ldots, d_k\}) \enspace .$$

Therefore $$ d \sqsubseteq (e \circ \delta) \sqcap (e \circ \infi \{d_1, \ldots, d_k\})$$ and from $d \neq \bot$ we obtain $(e \circ \delta) \neq \bot$ which implies $\delta \neq \bot$ due to axiom 2.(c) in Definition \ref{defQD}. This completes the proof.
\end{proof}

\subsection{ Comparison to Related Approaches} \label{sec:RA}

Other program transformations have been proposed in the literature with the aim of supporting $\simrel$-based reasoning while avoiding explicit $\simrel$-based unification. Here we draw some comparisons between the program transformation $\trans{\simrel}{\Prog}$ presented in the previous subsection, the program transformations $\extended{\lambda}{\Prog}$ and $\abstracted{\Prog}{\lambda}$ proposed  in \cite{Ses02}, and the program transformation $\MAProg$ proposed in \cite{MOV04}. These three transformations are applied to a classical logic program $\Prog$ w.r.t. a fuzzy similarity relation $\simrel$ over symbols in the program's signature. Both $\extended{\lambda}{\Prog}$ and $\abstracted{\Prog}{\lambda}$ are classical logic programs to be executed by $SLD$ resolution, and their construction depends on a fixed similarity degree $\lambda \in (0,1]$. On the other hand, $\MAProg$ is a multi-adjoint logic program over a particular multi-adjoint lattice $\GL$, providing the uncertain truth values in the interval $[0,1]$ and two operators for conjunction and disjunction in the sense of G\"{o}del's fuzzy logic (see \cite{Voj01} for technical details). As in the case of our own transformation $\trans{\simrel}{\Prog}$, the construction of $\MAProg$  does not depend on any fixed similarity degree. The transformation $\trans{\simrel}{\Prog}$ proposed in this paper is more general in that it can be applied to an arbitrary $\sqlp{\simrel}{\qdom}$ program
$\Prog$, yielding a $\qlp{\qdom}$ program $\trans{\simrel}{\Prog}$ whose least Herbrand model is the same  as that of $\Prog$.

We will restrict our comparisons  to the case that $\Prog$ is chosen as a similarity-based logic program in the sense of \cite{Ses02}. As an illustrative example, consider  the simple logic program $\Prog$ consisting of the following four clauses:
\begin{itemize}
    \item $C_r:\,\, r(X,Y) \gets p(X),\, q(Y),\, s(X,Y)$
    \item $C_p:\,\, p(c(U)) \gets$
    \item $C_q:\,\, q(d(V)) \gets$
    \item $C_s:\,\, s(Z,Z) \gets$
\end{itemize}

Assume an admissible similarity relation defined by $\simrel(c,d) = 0.9$ and consider the goal $G:\,\, \gets r(X,Y)$ for $\Prog$. Then, $\simrel$-based $SLD$-resolution as defined in \cite{Ses02} computes the answer substitution $\sigma = \{X \mapsto c(U),\, Y \mapsto d(U)\}$ with similarity degree $0.9$. This computation succeeds because $\simrel$-based unification can compute the $m.g.u.$ $\{Z \mapsto c(U),\, V \mapsto U\}$ with similarity degree $0.9$ to unify the two atoms $s(c(U),d(V))$ and $s(Z,Z)$. Let us now examine the behavior of the the transformed programs $\extended{0.9}{\Prog}$, $\abstracted{\Prog}{0.9}$, $\trans{\simrel}{\Prog}$ and $\MAProg$ and when working to emulate this computation without  explicit use of a $\simrel$-based unification procedure.

\begin{enumerate}

    \item $\extended{0.9}{\Prog}$ is defined in \cite{Ses02} as the set of all clauses $C'$ such that $\simrel(C,C') \geq 0.9$ for some clause $C \in \Prog$. In this case $\extended{0.9}{\Prog}$ includes the four clauses of $\Prog$ and the two additional clauses $p(d(U)) \gets$ and $q(c(V)) \gets$, derived by similarity from $C_p$ and $C_q$, respectively. Solving $G$ w.r.t. $\extended{0.9}{\Prog}$ by means of classical $SLD$ resolution produces two possible answer substitutions, namely $\sigma_1 = \{X \mapsto c(U),\, Y \mapsto c(U)\}$ and $\sigma_2 = \{X \mapsto d(U),\, Y \mapsto d(U)\}$. They are both similar to $\sigma$ to a degree greater or equal than $0.9$, but none of them is $\sigma$ itself, contrary to the claim in Proposition 7.1 (i) from \cite{Ses02}. Therefore, this Proposition seems to hold only in a somewhat weaker sense than the statement in \cite{Ses02}. This problem is due to the possible non-linearity of a clause's head, which is properly taken into account by our transformation $\trans{\simrel}{\Prog}$.

    \item According to \cite{Ses02}, $\abstracted{\Prog}{0.9}$ is computed from $\Prog$ by replacing all the constructor and predicate symbols by new symbols that represent the equivalence classes of the original ones modulo $\simrel$-similarity to a degree greater or equal than $0.9$. In our example these classes are $\{r\}$, $\{p\}$, $\{q\}$, $\{s\}$ and $\{c,d\}$, that can be represented by the symbols $r$, $p$, $q$, $s$ an $e$, respectively. Then, $\abstracted{\Prog}{0.9}$ replaces the two clauses $C_p$ and $C_q$ by $p(e(U)) \gets$ and $q(e(V)) \gets$, respectively, leaving the other two clauses unchanged. Solving $G$ w.r.t. $\abstracted{\Prog}{0.9}$ by means of classical $SLD$ resolution produces the answer substitution $\sigma' = \{X \mapsto e(U),\, Y \mapsto e(U)\}$, which corresponds to $\sigma$ modulo the replacement of the symbols in the original program by their equivalence classes. This is consistent with the claims in Proposition 7.2 from \cite{Ses02}.

    \item Note that $\Prog$ can be trivially converted into a semantically equivalent a $\sqlp{\simrel}{\U}$ program, just by replacing each occurrence of the implication sign $\gets$ in $\Prog$'s clauses by $\qgets{1.0}$. Then $\trans{\simrel}{\Prog}$ can be built as a $\qlp{\U}$ program by the method explained in Subsection \ref{sec:PT}. It includes three clauses corresponding to $C_r$, $C_p$ and $C_q$ of $\Prog$ plus the following three new clauses:
        \begin{itemize}
            \item $C_p':\,\, p(d(U)) \qgets{1.0} pay_{0.9}$
            \item $C_q':\,\, q(c(V)) \qgets{1.0} pay_{0.9}$
            \item $C_s':\,\, s(Z_1,Z_2) \qgets{1.0} Z_1 \sim Z_2$
        \end{itemize}
        where $C_p'$ resp. $C_q'$ come from replacing the linear heads of $C_p$ resp. $C_q$ by similar heads, and $C_s$ comes from linearizing the head of $C_s$, which allows no replacements by similarity. $\trans{\simrel}{\Prog}$ includes also the proper clauses for $\Prog_\sim$ and $\Prog_{\mathrm{pay}}$, in particular the following three ones:
    \begin{itemize}
        \item $I:\,\, X \sim X \qgets{1.0}$
        \item $S:\,\, c(X_1) \sim d(Y_1) \qgets{1.0} pay_{0.9}, X_1 \sim Y_1$
        \item $P:\,\, pay_{0.9} \qgets{0.9}$
    \end{itemize}
    Solving goal $G$ w.r.t. $\trans{\simrel}{\Prog}$ by means of the $\U$-qualified $SLD$ resolution procedure described in \cite{RR08} can compute the answer substitution $\sigma$ with qualification degree $0.9$. More precisely, the initial goal can be stated as $r(X,Y)\#W \sep W \geq 0.9$, and the computed answer is $(\sigma, \{W \mapsto 0.9\})$. The computation emulates $\simrel$-based unification of $s(c(U),c(V))$ and $s(Z,Z)$ to the similarity degree $0.9$ by solving $s(c(U),c(V))$ with the clauses $C_s'$, $I$, $S$ and $P$.

    \item The semantics of the \MALP framework depending on the chosen multi-adjoint lattice is presented in \cite{MOV04}. A comparison
with the semantics of the $\qlp{\qdom}$ scheme (see \cite{RR08} and Subsection \ref{QLP} above) shows that \MALP programs over the multi-adjoint lattice $\GL$ behave as $\qlp{\U'}$ programs, where $\U'$ is the quasi qualification domain analogous to $\U$ introduced at the end of Subsection \ref{QD} above. For this reason, we can think of the transformed program $\MAProg$ as presented with he syntax of a $\qlp{\U'}$ program. The original program $\Prog$ can also be  written as a $\qlp{\U'}$ program just by replacing each the implication sign $\gets$ occurring  in $\Prog$ by $\qgets{1.0}$. As explained in \cite{MOV04}, $\MAProg$ is built by extending $\Prog$ with clauses for  a new binary predicate $\sim$ intended to emulate the behaviour of $\simrel$-based unification between terms. In our example, $\MAProg$ will include (among others) the following clause
for $\sim$:
    \begin{itemize}
        \item $S':\,\, c(X_1) \sim d(Y_1) \qgets{0.9} X_1 \sim Y_1$
    \end{itemize}
In comparison to the clause $S$  in $\trans{\simrel}{\Prog}$, clause $S'$ needs no call to a $pay_{0.9}$ predicate at its body, because the similarity degree $0.9 = \simrel(c,d)$ can be attached directly to the clause's implication. This difference corresponds to the different interpretations of $\circ$, which behaves as $\times$ in $\U$ and as $min$ in $\U'$.

Moreover, $\MAProg$ is defined to include a clause of the following form for each pair of $n$-ary predicate symbols $pd$ and $pd'$ such that $\simrel(pd,pd') \neq 0$:
\begin{itemize}
    \item $C_{pd, pd'}:\,\, pd(Y_1, \ldots, Y_n)  \qgets{\simrel(pd,pd')} \\ pd'(X_1, \ldots, X_n), X_1 \sim Y_1,  \ldots, X_n \sim Y_n$
\end{itemize}
In our simple example, all the clauses of this form correspond to the trivial case where $pd$ and $pd'$ are the same predicate symbol and $\simrel(pd,pd') = 1.0$. Solving goal $G$ w.r.t.$\trans{\simrel}{\Prog}$ by means of the procedural semantics described in Section 4 of \cite{MOV04}
can compute the answer substitution $\sigma$ to the similarity degree $0.9$. More generally, Theorem 24 in \cite{MOV04} claims that for any choice of $\Prog$, $\MAProg$ can emulate any successful computation performed by $\Prog$ using $\simrel$-based $SLD$ resolution.
\end{enumerate}

In conclusion, the main difference between $\trans{\simrel}{\Prog}$  and $\MAProg$  pertains to the techniques used by both program transformations
in order to emulate the effect of replacing the head of a clause in the original program by a similar one. $\MAProg$ always relies on the clauses of the form $C_{pd, pd'}$ {\em and}  the clauses for $\sim$, while $\trans{\simrel}{\Prog}$ can avoid to use the clauses for $\sim$ as long as all the clauses involved in the computation have linear heads. In comparison to the two transformations $\extended{\lambda}{\Prog}$ and $\abstracted{\Prog}{\lambda}$, our transformation $\trans{\simrel}{\Prog}$ does not depend on a fixed similarity degree $\lambda$ and does not replace the atoms in clause bodies by similar ones.

\subsection{A Goal Solving Example } \label{sec:GS}

\begin{figure}[ht]
\begin{center}
\tt
\begin{tabular}{|r@{\hspace{0.2cm}}l|}
\hline
&\\
\scriptsize  1 & wild(lynx) <-0.9- $pay_{1.0}$\\
\scriptsize  2 & wild(boar) <-0.9- $pay_{1.0}$\\
\scriptsize  3 & wild(snake) <-1.0- $pay_{1.0}$\\
\scriptsize  4 & wild(cat) <-0.9- $pay_{0.8}$\\
\scriptsize  5 & wild(pig) <-0.9- $pay_{0.7}$\\
&\\
\scriptsize  6 & farm(cow) <-1.0- $pay_{1.0}$\\
\scriptsize  7 & farm(pig) <-1.0- $pay_{1.0}$\\
\scriptsize  8 & farm(boar) <-1.0- $pay_{0.7}$ \\
\scriptsize  9 & farm(cat) <-0.8- $pay_{0.3}$ \\
\scriptsize 10 & farm(lynx) <-0.8- $pay_{0.3}$ \\
\scriptsize 11 & farm(snake) <-0.4- $pay_{0.3}$ \\
&\\
\scriptsize 12 & domestic(cat) <-0.8- $pay_{1.0}$ \\
\scriptsize 13 & domestic(snake) <-0.4- $pay_{1.0}$ \\
\scriptsize 14 & domestic(lynx) <-0.8- $pay_{0.8}$ \\
\scriptsize 15 & domestic(cow) <-1.0- $pay_{0.3}$ \\
\scriptsize 16 & domestic(pig) <-1.0- $pay_{0.3}$ \\
\scriptsize 17 & domestic(boar) <-1.0- $pay_{0.3}$ \\
&\\
\scriptsize 18 & intelligent(A) <-0.9- $pay_{1.0}$,domestic(A) \\
\scriptsize 19 & intelligent(lynx) <-0.7- $pay_{1.0}$ \\
\scriptsize 20 & intelligent(cat) <-0.7- $pay_{0.8}$ \\
&\\
\scriptsize 21 & pacific(A) <-0.9- $pay_{1.0}$,domestic(A) \\
\scriptsize 22 & pacific(A) <-0.7- $pay_{1.0}$,farm(A) \\
&\\
\scriptsize 23 & pet(A) <-1.0- $pay_{1.0}$,pacific(A),intelligent(A) \\
&\\
\scriptsize 24 & $pay_{1.0}$ <-1.0- \\
\scriptsize 25 & $pay_{0.8}$ <-0.8- \\
\scriptsize 26 & $pay_{0.7}$ <-0.7- \\
\scriptsize 27 & $pay_{0.3}$ <-0.3- \\
& \\
\hline
\end{tabular}
\end{center}
\caption{Example of transformed program. (Note: no clauses for $\sim$ are needed because the original program was left-linear). \label{fig:exampletrans}}
\end{figure}

In order to illustrate the use of  the transformed program $\trans{\simrel}{\Prog}$ for golving goals w.r.t. the original program $\Prog$, we consider the case where $\Prog$ is the $\sqlp{\simrel}{\U}$ program displayed in Figure \ref{fig:example}. The transformed program $\trans{\simrel}{\Prog}$ obtained by applying Definition \ref{def:trans} is shown in Figure \ref{fig:exampletrans}.  The following observations are useful to understand how the transformation has worked in this simple case:

\begin{itemize}
    \item The value $\top$ in the domain $\U$ corresponds to the real number $1$ and hence by reflexivity $\simrel(A,A) = 1$ for any atom in the signature of the program. Therefore, and as a consequence of Definition \ref{def:trans}, every clause in the original program gives rise to a clause in the transformed program with the same head and with the same body except for a new, first atom $pay_{1.0}$. For instance, clauses 1, 2 and 3 in Figure \ref{fig:exampletrans} correspond to the same clause numbers in Figure \ref{fig:example}.

    \item Apart of the clauses corresponding directly to the original clauses, the program of Figure \ref{fig:exampletrans} contains new clauses obtained by  similarity with some clause heads in the original program. For instance,  lines 4 and 5  are obtained by similarity with clauses at lines 1 and 2 in the original program, respectively. The  subindexes at literal $pay$ correspond to $\simrel(\mathtt{lynx}, \mathtt{cat})=0.8$, $\simrel(\mathtt{boar}, \mathtt{pig}) = 0.7$, respectively.

    \item Analogously, for instance the clause at line 10 (with head {\tt farm(lynx)}) is obtained by head-similarity with the clause of line 6 in the $\sqlp{\simrel}{\U}$ program (head {\tt domestic(cat)}), and the subindex at $pay$ is obtained from
 $$
 \begin{array}{ll}
 \simrel(\mathtt{domestic(cat)}, \mathtt{farm(lynx)}) &= \\
 \simrel(\mathtt{domestic}, \mathtt{farm}) \sqcap \simrel(\mathtt{cat}, \mathtt{lynx}) &= \\
                                    0.3 \sqcap 0.8 &= \\
                                    0.3
 \end{array}$$

    \item There is no clause for predicate $\sim$ since all the heads in the original program were already linear and therefore $\Prog_\sim$ can be left empty in practice.

    \item The clauses  for $pay$ correspond to the fragment $\Prog_{\mathrm{pay}}$ in Definition \ref{def:trans}.
\end{itemize}

In the rest of this subsection, we will show an execution for the goal \texttt{pet(A)\#W | W >= 0.50} over the program $\trans{\simrel}{\Prog}$ (see Figure \ref{fig:exampletrans}) with the aim of obtaining all those animals that could be considered a \texttt{pet} for at least a qualification value of $0.50$.

We are trying this execution in the prototype developed along with \cite{RR08} for the instances $\qlp{\U}$ and $\qlp{\W}$. Although this prototype hasn't been released as an integrated part of $\toy$, you can download\footnote{Available at: \texttt{http://gpd.sip.ucm.es/cromdia/qlpd}. There you will also find specific instructions on how to install and run it as well as text files with the program examples tried in here.} the prototype to try this execution. Please notice that the prototype does not automatically do the translation process from a given $\sqlp{\simrel}{\qdom}$ program $\Prog$ to its transformed program $\trans{\simrel}{\Prog}$, because it was developed mainly for \cite{RR08}. Therefore, the transformed program shown in Figure \ref{fig:exampletrans} has been computed manually.

We will start running $\toy$ and loading the $\qlp{\U}$ instance with the command \texttt{/qlp(u)}:

\begin{verbatim}
Toy> /qlp(u)
\end{verbatim}

\noindent this will have the effect of loading the \emph{Real Domain Constraints library} and the $\qlp{\U}$ library into the system, the prompt \texttt{QLP(U)>} will appear. Now we have to compile our example program (assume we have it in a text file called \texttt{animals.qlp} in \texttt{C:/examples/}) with the command \texttt{/qlptotoy} (this command will behave differently based on the actual instance loaded).

\begin{verbatim}
QLP(U)> /qlptotoy(c:/examples/animals)
\end{verbatim}

Note that we didn't write the extension of the file because it \emph{must} be \texttt{.qlp}. This will create the file \texttt{animals.toy} in the same directory as our former file. And this one will be an actual $\toy$ program. We run the program with \texttt{/run(c:/examples/animals)} (again without the extension --although this time we are assuming \texttt{.toy} as extension--) and we should get the following message:

\begin{verbatim}
PROCESS COMPLETE
\end{verbatim}

And finally we are set to launch our goal with the command \texttt{/qlpgoal}. The solutions found for this program and goal are:

\begin{verbatim}
QLP(U)> /qlpgoal(pet(A)#W | W>=0.50)
      { A -> cat,
        W -> 0.5599999999999999 }

sol.1, more solutions (y/n/d/a) [y]?
      { A -> cat,
        W -> 0.7200000000000001 }

sol.2, more solutions (y/n/d/a) [y]?
      { A -> lynx,
        W -> 0.5760000000000002 }

sol.3, more solutions (y/n/d/a) [y]?
      { A -> lynx,
        W -> 0.5760000000000002 }

sol.4, more solutions (y/n/d/a) [y]?
      no
\end{verbatim}

At this point and if you remember the inference we did in Example \ref{ex:least-model} for \texttt{pet(lynx)\#0.50}, we have found a better solution (as you can see there are two solutions for \texttt{lynx}, and this is due to the two different ways of proving \texttt{intelligent(lynx)}: \texttt{intelligent(lynx)\#0.7} using clause 19, and \texttt{intelligent (lynx)\#0.576} using clauses 18 and 14.

\section{Conclusions} \label{Conclusions}

Similarity-based $LP$ has been proposed in \cite{Ses02} and related works to enhance the $LP$ paradigm with a kind of approximate reasoning which supports flexible information retrieval applications, as argued in \cite{LSS04,MOV04}. This approach keeps the syntax for program  clauses as in classical $LP$, and supports uncertain reasoning by using a fuzzy similarity relation $\simrel$ between symbols in the program's signature. We have shown that similarity-based $LP$ as presented in \cite{Ses02} can be reduced to Qualified $LP$ in the $\qlp{\qdom}$ scheme introduced in \cite{RR08}, which supports  logic programming with attenuated program clauses over a parametrically given domain $\qdom$  whose elements  qualify logical assertions by measuring their closeness to various users' expectations. Using  generalized similarity relations taking values in the carrier set of an
arbitrarily  given qualification domain $\qdom$,  we  have extended $\qlp{\qdom}$ to a more expressive scheme $\sqlp{\simrel}{\qdom}$ with two parameters, for programming modulo $\simrel$-similarity with $\qdom$-attenuated Horn clauses. We have presented a declarative semantics for $\sqlp{\simrel}{\qdom}$ programs and a semantics-preserving program transformation which embeds $\sqlp{\simrel}{\qdom}$ into $\qlp{\qdom}$. As a consequence, the sound and complete procedure for solving goals in $\qlp{\qdom}$ by $\qdom$-qualified $SLD$ resolution and its implementation in the
$\toy$ system \cite{RR08} can be used to implement $\sqlp{\simrel}{\qdom}$ computations via the transformation.

Our framework is quite general due to the availability of different qualification domains, while the similarity relations proposed in \cite{Ses02} take fuzzy values in the interval $[0,1]$. In comparison to the multi-adjoint framework proposed in \cite{MOV04}, the $\qlp{\qdom}$ and $\sqlp{\simrel}{\qdom}$ schemes have a different motivation and scope, due to the differences between multi-adjoint algebras and qualification domains as algebraic structures. In contrast to the goal solving procedure used in the multi-adjoint framework, $\qdom$-qualified $SLD$ resolution does not
rely on costly computations of reductant clauses and has been efficiently implemented.

As future work, we plan to investigate an extension of the $\simrel$-based $SLD$ resolution procedure proposed in \cite{Ses02} to be used within the $\sqlp{\simrel}{\qdom}$ scheme, and to develop an extension of this scheme which supports lazy functional programming and constraint programming facilities. The idea of  similarity-based unification has been already applied in \cite{MP06a} to obtain an extension of {\em needed narrowing}, the main goal solving procedure of functional logic languages. As in the case of \cite{Ses02}, the similarity relations considered in \cite{MP06a} take fuzzy values in the real interval $[0,1]$.

\acks
The authors have been partially supported by the Spanish National Projects MERIT-FORMS (TIN2005-09027-C03-03) and PROME-SAS--CAM (S-0505/TIC/0407).

\bibliographystyle{plainnat}

\begin{thebibliography}{22}

\bibitem{Apt90}
K.R. Apt.
\newblock Logic programming.
\newblock In J. van Leeuwen, editor, {\em Handbook of Theoretical Computer Science}, volume B: Formal Models and Semantics, pages 493-574. Elsevier and The MIT Press, 1990.

\bibitem{AVE82}
K.R. Apt and M.H. van Emden.
\newblock Contributions to the theory of logic programming.
\newblock {\em Journal of the Association for Computing Machinery (JACM)}, 29(3):841-862, 1982.

\bibitem{AF99}
F. Arcelli and F. Formato.
\newblock Likelog: A logic programming language for flexible data retrieval.
\newblock In {\em Proceedings of the 1999 ACM Symposium on Applied Computing (SAC'99)}, pages 260-267, New York, NY, USA, 1999. ACM Press.

\bibitem{toy}
P. Arenas, A.J. Fern\'andez, A. Gil, F.J. L\'opez-Fraguas, M. Rodr\'iguez-Artalejo and F. S\'aenz-P\'erez.
\newblock $\mathcal{TOY}$, a multiparadigm declarative language. Version 2.3.1, 2007.
\newblock R. Caballero and J. S\'anchez (Eds.), available at \texttt{http://toy.sourceforge.net}.

\bibitem{FGS00}
F. Formato, G. Gerla and M.I. Sessa.
\newblock Similarity-based unification.
\newblock {\em Fundamenta Informaticae}, 41(4):393-414, 2000.

\vfill\eject

\bibitem{GS99}
G. Gerla and M.I. Sessa.
\newblock Similarity in logic programming.
\newblock In G. Chen, M. Ying and K. Cai, editors, {\em Fuzzy Logic and Soft Computing}, pages 19-31. Kluwer Academic Publishers, 1999.

\bibitem{KS92}
M. Kifer and V.S. Subrahmanian.
\newblock Theory of generalized annotated logic programs and their applications.
\newblock {\em Journal of Logic Programming}, 12(3\&4):335-367, 1992.

\bibitem{LSS04}
V. Loia, S. Senatore and M.I. Sessa.
\newblock Similarity-based SLD resolution and its role for web knowledge discovery.
\newblock {\em Fuzzy Sets and Systems}, 144(1):151-171, 2004.

\bibitem{MOV01a}
J. Medina, M. Ojeda-Aciego and P. Vojt\'a\v{s}.
\newblock Multi-adjoint logic programming with continuous semantics.
\newblock In T. Eiter, W. Faber and M. Truszczyinski, editors, {\em Logic Programming and Non-Monotonic Reasoning (LPNMR'01)}, volume 2173 of {\em LNAI}, pages 351-364. Springer-Verlag, 2001.

\bibitem{MOV01b}
J. Medina, M. Ojeda-Aciego and P. Vojt\'a\v{s}.
\newblock A procedural semantics for multi-adjoint logic programming.
\newblock In P. Brazdil and A. Jorge, editors, {\em Progress in Artificial Intelligence (EPIA'01)}, volume 2258 of {\em LNAI}, pages 290-297. Springer-Verlag, 2001.

\bibitem{MOV04}
J. Medina, M. Ojeda-Aciego and P. Vojt\'a\v{s}.
\newblock Similarity-based unification: A multi-adjoint approach.
\newblock {\em Fuzzy Sets and Systems}, 146:43-62, 2004.

\bibitem{MP06a}
G. Moreno and V. Pascual.
\newblock Programming with fuzzy logic and mathematical functions. In A.P.I. Bloch and A. Tettamanzi, editors, {\em Proceedings of the 6th International Workshop on Fuzzy Logic and Applications (WILF'05)}, volume 3849 of {\em LNAI}, pages 89-98. Springer-Verlag, 2006.

\bibitem{RR08TR}
M. Rodr\'iguez-Artalejo and C.A. Romero-D\'iaz.
\newblock A generic scheme for qualified logic programming (Technical Report SIC-1-08).
\newblock Technical Report, Universidad Complutense, Departamento de Sistemas Inform\'aticos y Computaci\'on, Madrid, Spain, 2008.

\bibitem{RR08}
M. Rodr\'iguez-Artalejo and C.A. Romero-D\'iaz.
\newblock Quantitative logic programming revisited.
\newblock In J. Garrigue and M. Hermenegildo, editors, {\em Functional and Logic Programming (FLOPS'08)}, volume 4989 of {\em LNCS}, pages 272-288. Springer-Verlag, 2008.

\bibitem{Ses01}
M.I. Sessa.
\newblock Translations and similarity-based logic programming.
\newblock {\em Soft Computing}, 5(2), 2001.

\bibitem{Ses02}
M.I. Sessa.
\newblock Approximate reasoning by similarity-based SLD resolution.
\newblock {\em Theoretical Computer Science}, 275(1\&2):389-426, 2002.

\bibitem{Sub87}
V.S. Subrahmanian.
\newblock On the semantics of quantitative logic programs.
\newblock In {\em Proceedings of the 4th IEEE Symposium on Logic Programming}, pages 173-182, San Francisco, 1987.

\bibitem{Sub88}
V.S. Subrahmanian.
\newblock Query processing in quantitative logic programming.
\newblock In {\em Proceedings of the 9th International Conference on Automated Deduction}, volume 310 of {\em LNCS}, pages 81-100, London, UK, 1988. Springer-Verlag.

\bibitem{Sub07}
V.S. Subrahmanian.
\newblock Uncertainty in logic programming: Some recollections.
\newblock {\em Association for Logic Programming Newsletter}, 20(2), 2007.

\bibitem{VE86}
M.H. van Emden.
\newblock Quantitative deduction and its fixpoint theory.
\newblock {\em Journal of Logic Programming}, 3(1):37-53, 1986.

\bibitem{VEK76}
M.H. van Emden and R.A. Kowalski.
\newblock The semantics of predicate logic as a programming language.
\newblock {\em Journal of the Association for Computing Machinery (JACM)}, 23(4):733-742, 1976.

\bibitem{Voj01}
P. Vojt\'a\v{s}.
\newblock Fuzzy logic programming.
\newblock {\em Fuzzy Sets and Systems}, 124:361:370, 2001.
\end{thebibliography}

\end{document}